\pdfoutput=1

\documentclass{IEEEtran}

\usepackage{authblk}
\usepackage{amssymb}
\usepackage{amsthm}
\usepackage{color}
\usepackage{setspace}
\usepackage{amsfonts}
\usepackage{amssymb}
\usepackage{amsmath}
\usepackage{booktabs}
\usepackage{dsfont}
\usepackage{enumerate}
\usepackage{times}
\usepackage{todonotes}
\usepackage{latexsym}
\usepackage{pgf,pgfarrows,pgfnodes,pgfautomata,pgfheaps}
\usepackage{tikz}
\usetikzlibrary{arrows,automata,fit,matrix}
\usepackage{color}
\usepackage{algorithm,algorithmic}
\usepackage{extarrows}
\usepackage{soul}
\usepackage{subfig}
\usepackage{graphicx}
\usepackage[normalem]{ulem}

\usepackage{epstopdf}

\newtheorem{proposition}{Proposition}
\newtheorem{theorem}{Theorem}

\newtheorem{corollary}{Corollary}

\newcommand{\act}[1]{\xlongrightarrow{#1}}          

\newcommand{\calH}{\mathcal{H}}

\newcommand{\calR}{\mathcal{R}}

\newcommand{\hE}{\hat{E}}
\newcommand{\hA}{\hat{A}}
\newcommand{\hx}{\hat{x}}
\newcommand{\hb}{\hat{b}}

\newcommand{\norm}[1]{\lVert #1 \rVert}
\newcommand{\xe}{\mathrm{e}}            		

\newcommand{\xm}{x^-}
\newcommand{\xp}{x^+}

\newcommand{\zm}{z^-}
\newcommand{\zp}{z^+}

\newcommand{\im}{i^-}
\newcommand{\ip}{i^+}

\newcommand{\Vin}{\ensuremath{v_\mathit{in}}}
\newcommand{\Vout}{\ensuremath{v_\mathit{out}}}
\newcommand{\Vinm}{\Vin^-}
\newcommand{\Vinp}{\Vin^+}

\newcommand{\Voutm}{\Vout^-}
\newcommand{\Voutp}{\Vout^+}

\newcommand{\RE}{\mathbb{R}}
\newcommand{\calD}{\mathcal{D}}

\newcommand{\rI}{%
  \textup{\uppercase\expandafter{\romannumeral 1}}%
}

\newcommand{\rII}{%
  \textup{\uppercase\expandafter{\romannumeral 2}}%
}

\newcommand{\dert}{\ensuremath{\partial_t}}

\allowdisplaybreaks[0]

\begin{document}

\title{From Electric Circuits to Chemical Networks}

\author[1]{Luca Cardelli}
\author[2]{Mirco Tribastone}
\author[3]{Max Tschaikowski}
\affil[1]{University of Oxford and Microsoft Research}
\affil[2]{IMT School for Advanced Studies Lucca}
\affil[3]{Technical University of Vienna}

\maketitle

\begin{abstract}
Electric circuits manipulate electric charge and magnetic flux via a small set of discrete components to implement useful functionality over continuous time-varying signals represented by currents and voltages. Much of the same functionality is useful to biological organisms, where it is implemented by a completely different set of discrete components (typically proteins) and signal representations (typically via concentrations). We describe how to take a linear electric circuit and systematically convert it to a chemical reaction network of the same functionality, as a dynamical system. Both the structure and the components of the electric circuit are dissolved in the process, but the resulting chemical network is intelligible. This approach provides access to a large library of well-studied devices, from analog electronics, whose chemical network realization can be compared to natural biochemical networks, or used to engineer synthetic biochemical networks.
\end{abstract}

\section{Introduction}

Living organisms perform a variety of functions that can be described in abstract terms as information processing and regulation. Analogies have been drawn between the biochemical reaction networks that perform such functions and electric circuits of similar nature~\cite{Arkin2000,DelVecchio2008}. The comparison is useful in systems biology, in trying to understand the function of natural systems, and in synthetic biology, in trying to engineer desired functionality in a biochemical context.

An obstacle to exploiting this analogy is that the fundamental components of biochemistry and electronics are very different: the phenomena of resistance, induction, and capacitance based on the interplay between electric and magnetic fields have no immediate parallel in terms of chemical concentrations. Hence, it is not clear how to take systematic advantage of engineering knowledge developed in electronics in understanding biochemical systems. Instead, the search for \emph{components} of biological network has progressed in different and certainly more appropriate directions~\cite{Hart8346,Milo824}.

Even within electronics, though, the precise nature of electric components is incidental to the desired function. For example, one may wish to filter the high frequencies of a signal represented by an oscillating voltage. There are countless electric circuits that can perform this function, based on different classes of components in different configurations. Some of those circuits are based on operational amplifiers, which are themselves built from a large network of components to perform an abstract function to which they are incidental. The nature of the fundamental components is inessential, as long as they can be combined to provide a wide range of functionality.

A common way to describe essential function, both in electronics and in biochemistry, is through a system of ordinary differential equations (ODEs). Once the function of a circuit is reduced to this form, it does not matter if the quantities represented are voltages or concentrations: it only matters the way in which they vary over time. Conversely, given an ODE system, one may ask the engineering question of how to realize a circuit (electronic or biochemical) that can perform that function. An early example of this inverse process is the synthesis of mechanical and then electric analog circuits from differential equations~\cite{BUSH1931447,doi:10.1002/sapm1941201337}.

Certain classes of polynomial ODE systems can be systematically turned into chemical reaction networks (CRNs) that obey the same kinetics~\cite{Hungarian}. Further, CRNs can be compiled systematically into a collection of molecules that can be engineered to obey the kinetics of those reactions~\cite{Soloveichik5393}. In this paper we wish to go one step further on the front end of this process. Taking advantage of the large libraries of known electric circuits, we wish to take an arbitrary (but linear, for now) electric circuit and show how to turn it into a set of molecules that obey the kinetics of the quantities in that circuit.

The first obstacle we need to confront is that even common electric circuits describe behaviors that go beyond ODEs. Algorithmic approaches for the analysis of linear circuits such as Modified Nodal Analysis produce, in general, systems of differential algebraic equations (DAEs), where the algebraic equations are induced by classical node analysis based on Kirchhoff laws~\cite{ho1975modified}. Hence we first reduce DAEs to ODEs, after which we can apply some further techniques. The second obstacle is to take an ODE that may be about voltages and currents, and turn it into a form that can be interpreted chemically. This means that each variable should only take nonnegative quantities (for concentrations), and that appropriate chemical reactions should be derived about those quantities.

This approach has a dual purpose. From a systems biology point of view, we may want to compare the CRNs derived from electric circuits to the ones occurring in nature. This might help elucidate the function of natural networks. Or, at least, it will provide a spectrum of possible chemical networks of known function, whose structure may not be obvious, therefore broadening our expectations of what is possible chemically. From a systems biology point of view, we take it as given that an abstract CRN can be turned into a collection of molecules. In fact, multiple target architectures are possible, from short oligonucleotides in solution~\cite{Soloveichik5393} to gene networks~\cite{Genelets}. The ability to generate molecular configurations from a vast existing library of (electric, or other) circuits is appealing for systematizing the generations of synthetic organisms. In extending the known techniques from ODEs to DAEs, we extend the scope of potential libraries we can draw from.

\paragraph*{Contributions} Our main result is a systematic technique that transforms linear DAE systems into CRNs. Our technique is, to the best of our knowledge, novel. DAE systems are either solved symbolically by relying on index reduction~\cite{doi:10.1137/0909014}, or numerically by relying on numerical methods that compute the trajectories for a given initial condition~\cite{KunkelMehrmann}. In contrast to our approach, index reduction introduces additional derivatives of signals that may not be available to the circuit, while numerical methods do not transform the DAEs into ODEs, which seems necessary for transforming DAE systems into CRNs. We analyze, as an example, an electric high pass filter and provide a chemical reaction network for it, whose function and architecture can be independently interpreted in a biological context.

\section{Outline of Methods}\label{sec_outline}

We start from a linear electric circuit composed of resistors, capacitors and inductors, with variables ranging over voltages and currents, and we systematically derive an equivalent CRN where chemical species (or more precisely their differences) approximate the trajectories of the original variables.

The basis for this process is the so-called Hungarian Lemma~\cite{Hungarian}, which provides a method for converting certain polynomial ODEs into CRNs by converting each monomial on the right hand side of a differential equation into a separate chemical reaction. Polynomial ODEs can represent, exactly, a much broader class of ODEs, including fractional, trigonometric, and exponential terms~\cite{10.1007/978-3-319-19249-9_23}, thus covering a broad range of chemical behavior including Hill kinetics~\cite{Cardelli2014}.

The Hungarian Lemma, however, has specific requirements. First, the concentrations of the chemical species must be nonnegative, while ordinary ODE variables, and in particular voltages and currents, may be negative. Second, if a monomial has a negative sign, then the differential variable on the left-hand side of the equation must appear as a factor in the monomial. This means, for example, that the ODE $\dert x = y$ (where the growth rate of $x$ is given by the concentration of $y$, with $\dert x$ denoting the time derivative of $x$) can be reduced to the reaction $y \rightarrow x + y$. And the ODE $\dert x = -xy$ can be reduced to the reaction $x+y \rightarrow y$. But the ODE $\dert x = -y$ (where the decrease in rate of $x$ is given by the concentration of $y$) cannot be reduced: for $x$ to decrease it must appear on the left-hand-side of a chemical reaction, which implies it should appear in a monomial for $\dert x$ by the law of mass action. An ODE system with no such forbidden negative monomial is said to be \emph{Hungarian}, and any Hungarian ODE system can be reduced to a CRN (although not uniquely) whose mass action kinetics reproduces the original ODE. We use a technique to reduce a polynomial non-Hungarian ODE to an Hungarian one in twice as many variables, thus allowing us to produce CRNs also for non-Hungarian ODEs. In the same step, we make all trajectories nonnegative so that they can be realized by chemical species.

Some simple electric circuits yield ODE systems that can be converted to CRNs as outlined. Pure resistor circuits yield simple algebraic equations. But more complex circuits yield general DAEs~\cite{ho1975modified}, which we must be prepared to handle. Our main technique applies to linear DAE systems of the form
\begin{equation}\label{eq:dae}
E \dert x = A x + Bu
\end{equation}
Here, $x \in \RE^n$ is the (column) vector of dependent variables, $E \in \RE^{n \times n}$ produces a linear combination of their derivatives, $A \in \RE^{n \times n}$ produces a linear combination of the variables, and the term $B \in \RE^{n \times m}$ is the input matrix, and $u \in \RE^m$ is the vector of inputs, such as voltage or current sources.

In general, inputs are assumed to be arbitrary, known functions of time. However, for the purposes of converting electric circuits into CRNs, it is necessary to appropriately encode also the inputs as chemical species. In this paper, we will assume that the input vector $u$ can be itself described as the solution of a system of equations. Specifically, we assume that it satisfies an affine ODE system of the form
\begin{equation}\label{eq:input}
\begin{pmatrix}
\dert u \\
\dert z
\end{pmatrix} = D \cdot
\begin{pmatrix}
u \\
z
\end{pmatrix} + d
\end{equation}
for some matrix $D \in \RE^{(m+k) \times (m+k)}$ and $\left( u(0),z(0) \right)^T, d \in \RE^{m+k}$. Intuitively, the input $u$ in (\ref{eq:dae}) is part of an ODE solution which may depend on auxiliary ODE variables $z$ that do not appear in the DAE. This is a rather general setting that allows us to  encode arbitrary time-varying inputs by approximating them with Fourier series, which can be expressed as solutions of linear ODEs (see Section~\ref{sec:example}).

Our technique converts the overall system (\ref{eq:dae})-(\ref{eq:input}) into an ODE system over the same variables, up to a controllable approximation. That ODE system can then be transformed into a CRN as discussed above.

\begin{figure}
\begin{center}
\includegraphics[scale=0.60]{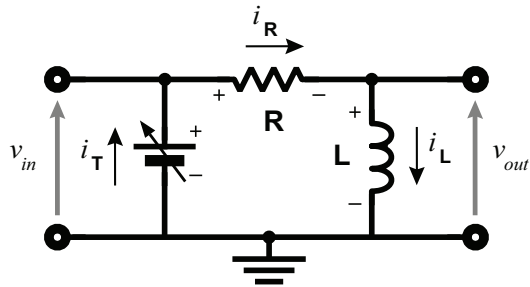}
\caption{\small \sl A high pass filter, with input voltage $v_{in}$ and output voltage $v_{out}$ with respect to ground.
         \label{fig:Figure1}}
\end{center}
\end{figure}

We now describe in detail the entire process of converting linear electric circuits to CRNs, through a small textbook example involving a single differential equation and a single algebraic equation for the well-known $RL$ (resistor-inductor) circuit in Figure \ref{fig:Figure1}~\cite{electric:circuit}. The analysis of the circuit proceeds as follows.
We let $\Vin$ denote the input voltage (measured with respect to the ground node); the output is the voltage $\Vout$ across the inductor.
By standard node analysis, using Kirchhoff's current law at each node, we obtain that the three currents are equal, so $i \triangleq i_T = i_R = i_L$. Faraday's law for the inductor $L$, and Ohm's law for the resistor $R$, then give us:
\begin{subequations}
    \begin{align}
    \dert i & = \Vout/L 		\label{DAE_1a} \\
    iR & = \Vin - \Vout 			\label{DAE_1b}
    \end{align}
\end{subequations}
This is a DAE, where (\ref{DAE_1a}) is a differential equation, and (\ref{DAE_1b}) is an algebraic equation. To find its solution, we can replace $\Vout = \Vin - iR$ from (\ref{DAE_1b}) in (\ref{DAE_1a}), obtaining (\ref{DAE_2a}), which can be integrated to obtain $i(t)$. From that solution we can then obtain $\Vout(t)$ via (\ref{DAE_2b}).
\begin{subequations}
    \begin{align}
    \dert i & = \Vin/L - iR/L		\label{DAE_2a} \\
    \Vout & = \Vin - iR				\label{DAE_2b}
    \end{align}
\end{subequations}

If we briefly assume that $\Vin$ is a non-negative input source, then (\ref{DAE_2a}) can be converted to a mass action CRN as follows (using chemical species with the same names as our variables):
    \begin{align*}
	\Vin  & \act{1/L}  i + \Vin &
	i & \act{R/L} \emptyset
    \end{align*}
where with the symbol $\emptyset$ we have denoted the empty set of products. These two chemical reactions yield (\ref{DAE_2a}) for the evolution of the concentration of species $i$. However, a voltage $\Vin$ may be negative, which cannot be modeled with chemical concentrations. Additionally, we are interested in the output $\Vout$, not $i$, and therefore we need to find a way to realize chemically equation (\ref{DAE_2b}) as well.

Because of those difficulties, we cannot make much progress without a more general technique to implement DAEs as CRNs. Our technique involves an approximation, like the one that seems necessary just for algebraic equations, but it can be used in the general case of linear DAEs. We now illustrate it by applying it to the example of Figure \ref{fig:Figure1}. We first rearrange our DAE as (\ref{DAE_3a},\ref{DAE_3b}). Setting  $x=(i, \Vout)^T$ we have:
\begin{subequations}
	\begin{align}
    \dert i & = \Vout/L				\label{DAE_3a} \\
    0 & = i + \Vout/R - \Vin/R			\label{DAE_3b}
    \end{align}
\end{subequations}
which can be arranged into the form (\ref{eq:dae}):
\begin{align*}
	\underbrace{
    \begin{pmatrix}
      1 & 0 \\
      0 & 0
    \end{pmatrix}
    }_{E}
    \underbrace{
    \begin{pmatrix}
      \dert i \\
      \dert \Vout
    \end{pmatrix}
	}_{\dert x}
    	& =
    \underbrace{
    \begin{pmatrix}
      0 & 1/L \\
      1 & 1/R
    \end{pmatrix}
	}_{A}
    \underbrace{
    \begin{pmatrix}
      i \\
      \Vout
    \end{pmatrix}
	}_{x}
    	+
    \underbrace{
    \begin{pmatrix}
      0 \\
      -\Vin/R
    \end{pmatrix}
	}_{b}
    \label{DAE_3c}
\end{align*}
where we have fixed the constant vector $b := Bu$ under the assumption of a constant input source $u := \Vin$.

We approximate the DAE, for a parameter $h>0$, by symbolically computing $F_h(x) = (E-hA)^{-1}(Ax+b)$, which is related to the numerical backward Euler method. Taking $R = L = 1$ for simplicity, we obtain:
    \begin{align}
    F_h
    \begin{pmatrix}
      i \\
      \Vout
    \end{pmatrix}
	    =
    \begin{pmatrix}
      \dfrac{\Vin-i}{1+h} \\
      \dfrac{\Vin-i-(1+h)\Vout}{h(1+h)}
    \end{pmatrix}
    \end{align}

We next use $F_h(x)$ as the right-hand side of a new ODE system $\dert x = F_h(x)$, which is such that for $h \rightarrow 0$ the solution of the ODE system (\ref{DAE_4a})-(\ref{DAE_4b}) converges to the solution of the original DAE system (\ref{DAE_3a})-(\ref{DAE_3b}) (see Theorem~\ref{thm_main}).
\begin{subequations}
    \begin{align}
        \dert i & = \Vin/(1+h) - i/(1+h)			    	\label{DAE_4a} \\
        \dert \Vout & = \Vin/(h+h^2) - i/(h+h^2) - \Vout/h	\label{DAE_4b}
    \end{align}
 \end{subequations}
Indeed, we can easily see that for $h \rightarrow 0$, (\ref{DAE_4a}) converges exactly to the ODE (\ref{DAE_2a}). As for (\ref{DAE_4b}), this is now a differential equation approximating equation (\ref{DAE_3b}) for $h \rightarrow 0$, where we notice that a small value of $h$ makes (\ref{DAE_4b}) evolve much faster than (\ref{DAE_4a}).


We have reduced a DAE to an ODE, but (\ref{DAE_4a})-(\ref{DAE_4b}) is not Hungarian because of the $-i$ monomial in (\ref{DAE_4b}). Keeping in mind that we need to deal eventually with non-Hungarian ODEs, we now apply a \emph{positivation} technique in the style of Oishi and Klavins~\cite{5979221}, where each variable is represented as the difference of two \emph{non-negative} variables:
    \begin{align*}
      i & = \ip - \im	& 	\Vin & = \Vinp - \Vinm	& \Vout  & = \Voutp - \Voutm
    \end{align*}

Let us now abbreviate $p=1/(1+h)$, $q=1/(h+h^2)$, and $r=1/h$, and consider the ODE system where we separate the positive and negative monomials of each ODE in (\ref{DAE_4a})-(\ref{DAE_4b}) into two ODEs:
\begin{subequations}
    \begin{align}
      \dert \ip  & =   p\Vinp + p\im & \dert \im & =   p\Vinm + p\ip			\label{DAE_5a} \\
      \dert \Voutp & = q\Vinp + q\im + r\Voutm & \dert \Voutm & = q\Vinm + q\ip + r\Voutp	\label{DAE_5b}
    \end{align}
 \end{subequations}

The initial conditions for this new system must satisfy $\ip_0 - \im_0 = i_0$ with $\ip_0,\im_0\geq 0$, etc. Since differentiation is a linear operator, the solutions of (\ref{DAE_4a})-(\ref{DAE_4b}) can be recovered as differences from the solutions of (\ref{DAE_5a})-(\ref{DAE_5b}): $\dert \ip - \dert \im = \dert i$ and $\dert \Voutp - \dert \Voutm = \dert \Vout$. Although the goal was to make all variables non-negative, we now also have a Hungarian ODE system because all the monomials in (\ref{DAE_5a})-(\ref{DAE_5b}) are positive. Hence there is no further difficulty in converting these ODEs to mass action reactions, obtaining the following linear CRN with one reaction for each monomial in (\ref{DAE_5a})-(\ref{DAE_5b}), and with the parameter $h$ appearing in the reaction rates:
 \begin{align}
\Vinp & \act{p} \Vinp + \ip		& \im & \act{p} \im + \ip		\nonumber \\
\Vinm & \act{p} \Vinm + \im		& \ip & \act{p} \ip + \im		 \nonumber \\
\Vinp & \act{q} \Vinp + \Voutp	& \im & \act{q} \im + \Voutp	\label{eq_hungarization_example}  \\
\Vinm & \act{q} \Vinm + \Voutm	& \ip & \act{q} \ip + \Voutm  \nonumber \\
\Voutm & \act{r} \Voutm + \Voutp	& \Voutp & \act{r} \Voutp + \Voutm \nonumber
    \end{align}
Here the input $\Vin^\pm$ always acts as a simple catalyst. We note that the CRN implementation does not depend on the actual value of $\Vin$, which only affects the initial condition of the chemical species that represent $\Vin^{\pm}$. This decoupling between the CRN implementation of the circuit and that of the input sources carries over to the more general case when the sources are time-varying solutions of the ODEs~(\ref{eq:input}), see Theorem~\ref{thm_main_ext_ext} and the subsequent discussion.

Inspecting (\ref{eq_hungarization_example}), the chemical species $\Vout^{\pm}$ and $i^{\pm}$ are involved in autocatalytic cycles. For example both $\ip$ and $\im$ grow exponentially over time, while their difference $i$ remains bounded. It is possible to eliminate such exponential growths by adding non-linear dampening reactions to the otherwise linear CRN:
    \begin{align}\label{eq_correction_example}
	\ip + \im & \act{\gamma} \emptyset &
    \Voutp + \Voutm & \act{\gamma} \emptyset
    \end{align}
The first reaction, for example, preserves the difference $\ip - \im$, and results in two new identical monomials in the ODEs for $\ip$ and $\im$, that then cancel in $\dert \ip - \dert \im$. Hence that reaction does not change the $i$ solution, but keeps $\ip$ and $\im$ bounded.

The network consisting of reactions (\ref{eq_hungarization_example})-(\ref{eq_correction_example}) is depicted in Figure \ref{fig:Figure2}, where for small $h$ we have $1 \approx p \ll q \approx r$, and we can take $\gamma = r$.
This network has the flavor of an incoherent feedforward motif \cite{Milo824}, considering  parallel pairs $x^\pm \to y^\pm$ as activations and cross pairs $x^\pm \to y^\mp$ as inhibitions, thereby $\Vin$ activates both $i$ and $\Vout$, and $i$ `incoherently' inhibits $\Vout$. Additionally, the motif of mutual catalysis and join degradation around $i^\pm$ makes that pair stabilize to a copy of its input $\Vin^\pm$ (in the sense that at steady state $\ip - \im = \Vinp - \Vinm$) regardless of the value of the rate $p$. This motif is repeated around $\Vout^\pm$.
%
%
%
%
%
When the input $\Vin^\pm$ remains constant, $i^\pm$ becomes a copy of $\Vin^\pm$, and $\Vout^\pm$ becomes a copy of the sum of its two opposite inputs, $\Vin^\pm$ and $i^\mp$, and so it converges to a baseline output of $\Voutp - \Voutm \approx 0$. When the input $\Vin$ changes, it affects $\Vout$ quickly and $i$ slowly, with a delayed inhibition of $\Vout$ by $i$. It has been shown that feedforward motifs can behave like high-pass filters \cite{deRonde2012}.

The subnetwork in Figure~\ref{fig:Figure2} consisting of $\Vin^\pm$, $\Vout^\pm$, and the connecting $q$,$r$,$\gamma$ arcs, is in itself also a low-pass filter. It is exactly what is obtained when replacing the inductor with a capacitor of capacitance $C$ in Figure \ref{fig:Figure1}, yielding a well known low-pass filter, and deriving the CRN from it by positivation (with $q=r=1/RC$). The process is simpler in this case, since a single ODE is generated from that circuit, and no approximation via $h \rightarrow 0$ is required.

In summary, we have derived an intelligible chemical reaction network from an electric circuit, and we are guaranteed that it implements the same functionality, as shown in the next section.
\begin{figure}
\begin{center}
\includegraphics[scale=0.70]{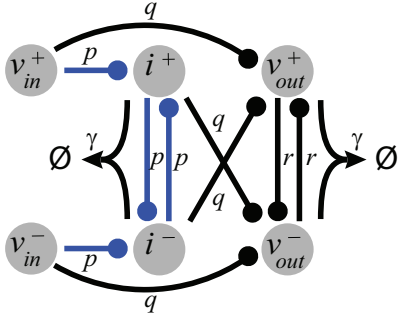}
\caption{\small \sl The CRN for the high-pass filter of Figure \ref{fig:Figure1}, consisting of reactions (\ref{eq_hungarization_example}, \ref{eq_correction_example}). A ball-headed arc from $x$ to $y$ denotes a reaction $x \to x + y$, and a double-tailed arc from $x$ and $y$ denotes a reaction $x + y \to \emptyset$. Black arcs have much faster rates than blue arcs.} \label{fig:Figure2}
\end{center}
\end{figure}

\section{Methods}\label{sec_methods}


\subsection{From DAEs to ODEs}\label{sec_dae_2_ode}
We first show how to convert a linear DAE system into a linear ODE system. To this end we start with a DAE system in the form
\begin{equation}\label{eq:constant.dae}
E \dert x  = A x + b, \qquad \text{with~} E, A \in \RE^{n \times n}, b \in \RE^n,
\end{equation}
which corresponds to (\ref{eq:dae}) under the assumption of constant inputs $u$. From now on we focus on \emph{regular} DAE systems, i.e., systems for which there exists no initial condition $x(0) \in \RE^n$ for which they admit more than one solution. We let $\calD$ denote the set of initial conditions for which a regular DAE admits solutions. Elements of $\calD$ are called \emph{consistent initial conditions} (see, e.g.,~\cite[Section 2.1]{KunkelMehrmann} for details).
Obviously, any physically meaningful DAE model has to be regular. In the case of electric circuits, for instance, non-regular DAE systems may arise in the presence of short circuits and other erroneous designs.

If $E$ is invertible, this DAE can be directly recast into a linear ODE system via $$\dert x = E^{-1} A x + E^{-1} b.$$
However, in the case of linear electric circuits $E$ is in general  not invertible. In this case, a transformation of a DAE system into an ODE system requires \emph{index reduction}~\cite{doi:10.1137/0909014,KunkelMehrmann}, which relies on expensive symbolic computations. Our method consists in circumventing this analysis by considering an explicit scheme arising from numerical methods for the solution of DAE systems.


A numerical method is an algorithm that generates, for a given small time step $h > 0$ and initial condition $x(0)$, a sequence $(x[i])_i$ such that $x[i] \approx x(i h)$, where $x : [0;\infty) \to \RE^n$ denotes the true solution of (\ref{eq:constant.dae}). Numerical methods are guaranteed to converge to the true solution $x$ when $h$ approaches zero. That is, that for any error threshold $\varepsilon > 0$ and finite time horizon $T > 0$, one can find a sufficiently small time step $h > 0$ such that $\max_{0 \leq i \leq N} \norm{x[i] - x(i h)} \leq \varepsilon$ and $T / h = N \in \mathbb{N}$.

A common numerical method for the solution of DAE systems is the backward Euler method~\cite[Section 5.2]{KunkelMehrmann} which is given by
$$x[i + 1] = x[i] + h F_h(x[i]), \text{with~} F_h(x) := (E - hA)^{-1}(A x + b).$$
The function $F_h$ is well-defined for sufficiently small $h > 0$ because the DAE system is regular~\cite[Section 2.1]{KunkelMehrmann}.

Noting that the ``slope'' of the Euler method at point $x[i]$, $(x[i + 1] - x[i]) / h$, is given by $F_h(x[i])$, we expect that the Euler sequence $(x[i])_{i \geq 0}$ will match the solution of the ODE system $\dot{x}_h = F_h(x_h)$. That is, we expect that $x_h(i h) \approx x[i]$ for all $0 \leq i \leq N$. Then, the convergence of the Euler sequence to the DAE solution $x$ would allow us to conclude that $x_h(i h) \approx x[i] \approx x(i h)$.

The next theorem is our first main result and proves that this is indeed the case.

\begin{theorem}\label{thm_main}
For any $\varepsilon > 0$, $x(0) \in \calD$ and $T > 0$, there exists an $h > 0$ such that $\sup_{0 \leq t \leq T} \norm{x(t) - x_h(t)} \leq \varepsilon$, where $\dert x_h = F_h(x_h)$ and $x_h(0) = x(0)$.
\end{theorem}

We now extend Theorem~\ref{thm_main} to systems (\ref{eq:dae})-(\ref{eq:input}).

\begin{theorem}\label{thm_main_ext}
Given a DAE system via (\ref{eq:dae}) and (\ref{eq:input}), consider the ODE system
\begin{align}\label{eq_ode_approx_compositional}
\dert x_h = (E - hA)^{-1}\left(A x_h + B u_h^{\langle 0 \rangle} + h B u_h^{\langle 1 \rangle} \right) ,
\end{align}
where $h > 0$ is small and the functions $u_h^{\langle 0 \rangle}, u_h^{\langle 1 \rangle} \in \RE^m$ satisfy
\begin{align}
\begin{pmatrix}
\dert u^{\langle 0 \rangle}_h \\
\dert z^{\langle 0 \rangle}_h
\end{pmatrix} & =
(I - h D)^{-1} D
\begin{pmatrix}
u_h^{\langle 0 \rangle} \\
z_h^{\langle 0 \rangle}
\end{pmatrix} + (I - h D)^{-1} d \label{eq_crn_1} \\
\begin{pmatrix}
\dert u^{\langle 1 \rangle}_h\\
\dert z^{\langle 1 \rangle}_h
\end{pmatrix}
& = (I - h D)^{-1} D
\begin{pmatrix}
u_h^{\langle 1 \rangle}\\
z_h^{\langle 1 \rangle}\\
\end{pmatrix}\label{eq_crn_2}
\end{align}
with initial conditions
\begin{align*}
u_h^{\langle 0 \rangle}(0) & = u(0) &
u_h^{\langle 1 \rangle}(0) & = (I - h D)^{-1} D u(0) \\
z_h^{\langle 0 \rangle}(0) & = z(0) &
z_h^{\langle 1 \rangle}(0) & = (I - h D)^{-1} D z(0)
\end{align*}
Then, for any $\varepsilon > 0$, $x(0) \in \calD$ and $T > 0$, there exists an $h > 0$ such that $\sup_{0 \leq t \leq T} \norm{x(t) - x_h(t)} \leq \varepsilon$ if $x_h(0) = x(0)$.
\end{theorem}

Note that $I - h D$ is strictly diagonal dominant and therefore invertible for sufficiently small values of $h$. It can be shown that $(u^{\langle 0 \rangle}_h, v^{\langle 0 \rangle}_h)^T$ converges to $(u,v)^T$ from~(\ref{eq:input}) as $h \to 0$. Hence, Theorem~\ref{thm_main_ext} essentially replaces the constant vector $b$ in $F_h(x) = (E - hA)^{-1}(A x + b)$ by the function $B u$.


\subsection{From ODEs to CRNs}\label{sec_odes_2_crns}

We next present a technique that transforms the ODE approximation from Section~\ref{sec_dae_2_ode} into a CRN. The approach borrows ideas from~\cite{5979221} that transforms linear ODE systems (i.e., systems \emph{without} algebraic constraints) into CRNs. We wish to point out, however, that our approach considers state space representation, while~\cite{5979221} acts on the frequency domain.

In the following, let $\dert x = \hA x + \hb$ denote some ODE system with initial condition $x(0)$.
\begin{proposition}\label{prop_positivation}
Any non-negative quadruple $(\hA^+,\hA^-,\hb^+,\hb^-)$ satisfying $\hA = \hA^+ - \hA^-$ and $\hb = \hb^+ - \hb^-$ induces the positivation
\begin{equation}\label{eq_positivation}
\begin{split}
\dert x^+ & = \hA^+ x^+  +   \hA^- x^-  +  \hb^+ \\
\dert x^- & = \hA^+ x^-  +   \hA^- x^+  +  \hb^-
\end{split}
\end{equation}
If~(\ref{eq_positivation}) is subject to non-negative $x^+(0), x^-(0) \in \RE^n_{\geq0}$ with $x(0) = x^+(0) - x^-(0)$, the solution $(x^+, x^-)$ remains non-negative and satisfies $x = x^+ - x^-$.
\end{proposition}

While positivations trivially satisfy the properties of the Hungarian lemma discussed in Section~\ref{sec_outline} and can therefore be readily translated into CRNs, they may exhibit divergence even if the original system is bounded, see for instance~(\ref{DAE_5a}), which implies that $i^+$ and $i^-$ diverge. Fortunately, one can apply a correction that leads to bounded positivations.

\begin{proposition}\label{prop_bounded}
Given a positivation $(\hA^+,\hA^-,\hb^+,\hb^-)$, define the quadratic function $Q(x^+,x^-) = (x_1^+ x_1^-,\ldots,x_n^+ x_n^-)^T$. Then, for any $\gamma > 0$, the corresponding Hungarization is
\begin{equation}\label{eq_hungarization}
\begin{split}
\dert x^+ & = \hA^+ x^+  +   \hA^- x^-  +  \hb^+ - \gamma Q(x^+,x^-) \\
\dert x^- & = \hA^+ x^-  +   \hA^- x^+  +  \hb^- - \gamma Q(x^+,x^-)
\end{split}
\end{equation}
If~(\ref{eq_hungarization}) is subject to non-negative $x^+(0), x^-(0) \in \RE^n_{\geq0}$ with $x(0) = x^+(0) - x^-(0)$, ODE system~(\ref{eq_hungarization}) admits a non-negative solution on $[0;\infty)$ that satisfies $x = x^+ - x^-$. Moreover, if $x$ is bounded on $[0;\infty)$, then so is $(x^+, x^-)$.
\end{proposition}

By applying the law of mass action, it can be easily seen that the $Q$ terms in~(\ref{eq_hungarization}) are captured by the annihilation reactions $\xp_1 + \xm_1 \act{\gamma} \emptyset,  \ldots, \xp_n + \xm_n \act{\gamma} \emptyset$. Combining this with Proposition~\ref{prop_bounded}, we arrive at the following statement.

\begin{proposition}\label{prop_crn_encoding}
Define the chemical reactions of Hungarization~(\ref{eq_hungarization}) as
\begin{align*}
\calR & = \{ x^+_i + x^-_i \act{\gamma} \emptyset \mid i \} \cup \{ \act{\hb^+_i} x^+_i \mid i \} \cup \{ \act{\hb^-_i} x^-_i \mid i \} \\
& \cup \{ x^+_j \act{\hA^+_{i,j}} x^+_j + x^+_i \mid i,j \} \cup \{ x^-_j \act{\hA^-_{i,j}} x^-_j + x^+_i \mid i,j \} \\
& \cup \{ x^+_j \act{\hA^-_{i,j}} x^+_j + x^-_i \mid i,j \} \cup \{ x^-_j \act{\hA^+_{i,j}} x^-_j + x^-_i \mid i,j \}
\end{align*}
Then, reactions $\calR$ induce, via the law of mass action, the ODE system~(\ref{eq_hungarization}).
\end{proposition}

For instance, if the positivation is given by~(\ref{DAE_4a}-\ref{DAE_4b}), then~(\ref{eq_hungarization_example}-\ref{eq_correction_example}) constitute $\calR$.

Proposition~\ref{prop_crn_encoding} and Theorem~\ref{thm_main_ext} yield our main result.

\begin{theorem}\label{thm_main_ext_ext}
Given a DAE system via (\ref{eq:dae}) and (\ref{eq:input}), let
\begin{itemize}
    \item $\calH_{1,h}$ and $\calH_{2,h}$ denote the Hungarization of~(\ref{eq_ode_approx_compositional}) and~(\ref{eq_crn_1}-\ref{eq_crn_2}), respectively.
    \item $\calR_{1,h}$ and $\calR_{2,h}$ refer to the chemical reactions of $\calH_{1,h}$ and $\calH_{2,h}$, respectively.
\end{itemize}
Then, the following holds true.
\begin{enumerate}[a)]
    \item The solution of the union CRN given by $\calR_{1,h} \cup \calR_{2,h}$ converges to the DAE solution as $h \to 0$.
    \item A change of $D$ and $d$ affects the reactions $\calR_{2,h}$ but does not alter the reactions $\calR_{1,h}$.
\end{enumerate}
\end{theorem}

As anticipated in Section~\ref{sec_outline}, Theorem~\ref{thm_main_ext_ext} ensures that a) our encoding is correct up to a controllable error and b) that the CRN implementation of the circuit, $\calR_{1,h}$, does not depend on the CRN implementation of the input, $\calR_{2,h}$.

Theorems \ref{thm_main_ext} and \ref{thm_main_ext_ext} allow for composition of circuits: a circuit expressed as a DAE (\ref{eq:dae}), with input provided by and ODE (\ref{eq:input}), yields another ODE (\ref{eq_ode_approx_compositional}) that can be supplied as input to a further circuit. The corresponding CRNs can be composed as well.

\section{Methods Applied to Example}\label{sec:example}

The $RL$ circuit discussed in Section~\ref{sec_outline} is a high-pass filter, attenuating the low frequencies of the input while transmitting the high frequencies to the output. The \emph{cutoff frequency} is the frequency at which the input signal is attenuated by $\frac{1}{2}$ its power, or equivalently its amplitude is attenuated by $-3\text{dB} \approx \sqrt{\frac{1}{2}} \approx 0.707$. In our circuit, the cutoff frequency is $f_c = \frac{R}{2\pi L} \text{Hz}$, where $R$ is the value of the resistance (measured in ohm) and $L$ is the inductance (in henry).

Here we show how we can apply Theorem~\ref{thm_main_ext_ext} to the $RL$ circuit with a time-varying input represented by an arbitrary differentiable function. Such a function can be approximated arbitrarily well by a Fourier series $u(t)$ given by \begin{align*}
u(t) & = \alpha + \sum_{i = 1}^N \beta_i \sin(\omega_i t + \gamma_i) ,
\end{align*}
where $\alpha$, $\beta_i$, $\omega_i$, and $\gamma_i$ are constant parameters. It can be seen that $u(t)$ can be written as the solution of the ODE system:
\begin{align}
\partial_t u & = \sum_{i = 1}^N \beta_i\omega_i \bar{z}_i &
\partial_t z_i & = \omega_i \bar{z}_i &
\partial_t \bar{z}_i & = - \omega_i z_i \label{eq:fourier}
\end{align}
with initial conditions $z_i(0) = \sin(\gamma_i)$, $\bar{z}_i(0) = \cos(\gamma_i)$ and $u(0) = \alpha + \sum_{i = 1}^N \beta_i \sin(\gamma_i)$ with $1 \leq i \leq N$, whereby $z_i$ and $\bar{z}_i$ are  auxiliary variables whose solutions give the sinusoidal components of the series. We can recognize the system~(\ref{eq:fourier}) to be in the required form (\ref{eq:input}).

\begin{figure}
\begin{center}
\includegraphics[scale=1.0]{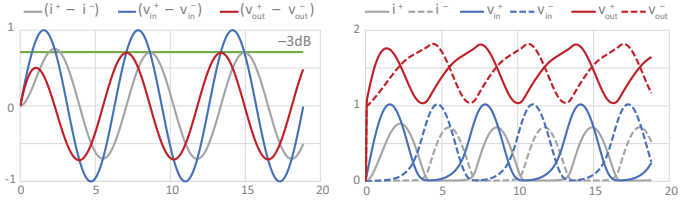}
\caption{\small \sl Left: simulation of the CRN from Figure \ref{fig:Figure2}, plotting variables differences, with $L=R=1$, input $\Vin = \Vinp - \Vinm$ of frequency $\frac{1}{2\pi} \text{Hz}$, and output $\Vout = \Voutp - \Voutm$. Horizontal axis is time. Right: the same, but plotting the individual variables.
         \label{fig:Figure3}}
\end{center}
\end{figure}
As a first example of an input waveform, we consider the ODE system
\begin{align} \label{eq_sin}
\begin{pmatrix}
\partial_t u \\
\partial_t z
\end{pmatrix}  =
\begin{pmatrix}
0 & 1 \\
-1 & 0 \\
\end{pmatrix} \cdot \begin{pmatrix}
u \\
z
\end{pmatrix}
\end{align}
With initial conditions $u(0) = 0$ and $z(0) = 1$, this yields the solution $u(t) = \sin(t)$ for all $t$, whose frequency is the cutoff frequency.
Figure \ref{fig:Figure3} shows simulations of the $\sin(t)$ CRN composed with the high-pass filter CRN, taking $h = 0.01$ for a sufficiently good approximation, and $\gamma = r$ for a sufficiently fast degradation. As expected, the output $\Vout$ is attenuated by $-3\text{dB} \approx 0.707$, and its phase is shifted by $45^{\circ}$. The variation of the underlying non-negative variables is shown on the right.

\begin{figure}
\begin{center}
\includegraphics[width=0.7\columnwidth]{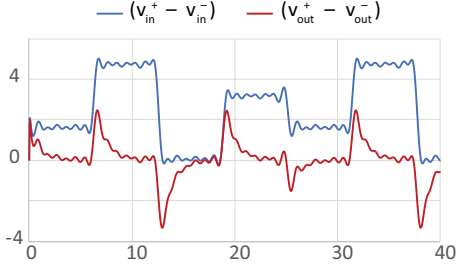}

\caption{\small \sl The circuit from Figure \ref{fig:Figure2} exhibiting perfect adaptation. Input $\Vin = \Vinp - \Vinm$, and output $\Vout = \Voutp - \Voutm$. Horizontal axis is time.
         \label{fig:Figure4}}
\end{center}
\end{figure}
As a second example, in a biological context, high pass filters exhibit \emph{perfect adaptation}~\cite{FerrellAdaptation}: they adapt to slow but possibly large variations in input stimulus and still react to quick changes. In Figure \ref{fig:Figure4} we supply stepped inputs via an appropriate Fourier series to our filter. At each sudden increase or decrease, the output reacts quickly and then settles back to its original level. The size of the transient response is proportional to the step size, but independent of the level of the input. The adaptation level can be set to any level, not just zero, by adding a constant contribution to $\Vout^+$, so that the output can represent the (positive) concentration of a certain protein.

\section{Discussion}
We have presented a method to convert linear DAEs to CRNs which hinges on a transformation into an approximate linear ODE system with arbitrary accuracy. This is then translated into a set of reactions where the time-course evolutions of the concentrations of the chemical species can be directly related to the original DAE solution.

In principle, any DAE system can be exactly transformed into an ODE system by means of so-called \emph{index reduction}~\cite{KunkelMehrmann}. However, this relies on symbolic computations. Moreover, the so-obtained ODE system will contain derivatives of the input signal, thus requiring for additional approximations. For instance, by differentiating (\ref{DAE_1b}) and using (\ref{DAE_1a}) we obtain $\partial_t \Vout = \partial_t \Vin - \frac{R}{L}\Vout$, which together with (\ref{DAE_1a}) is a simple ODE system in the dependent variables with no algebraic equations. This system now depends on the \emph{derivative} of the input, $\partial_t \Vin$, and would have to be combined with another circuit to supply that derivative from the given input $\Vin$. In contrast to index reduction our technique does not requires  input derivatives. Moreover, while index reduction requires symbolic computations, the matrix inversion at the basis of the construction of the approximate linear ODE system can also be performed using numerical techniques.

\begin{figure}
\begin{center}
\includegraphics[scale=0.70]{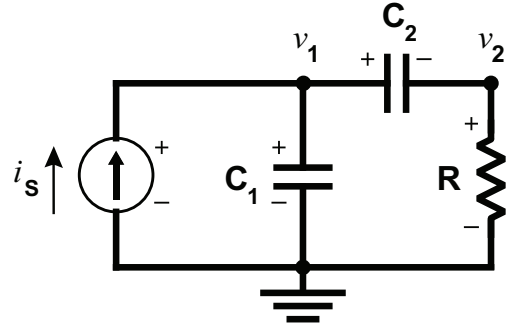}
\caption{\small \sl An electric circuit which gives the DAE system (\ref{eq_non_semi_explicit}) that is not in semi-explicit form.
         \label{fig:Figure5}}
\end{center}
\end{figure}

The precision of the linear ODE depends on a parameter which, when taken asymptotically small, has the effect of rapidly equilibrating certain components of the ODE system. Therefore, in this respect our approach can be related to quasi steady-state approximation  (QSSA,~\cite{VerhulstBook,Pantea2014}), which applies to \emph{semi-explicit} DAEs in the special form \begin{equation*}
\begin{split}
\partial_t x & = A_1 x + B_1 u \\
0 & = A_2 y + B_2 u\\
\end{split}
\end{equation*}
Essentially, QSSA replaces the algebraic constraints $0 = A_2 y + B_2 u$ with $\varepsilon \partial_t y = A_2 y + B_2 u$ for some $\varepsilon \approx 0$. However, DAEs of electric circuits are not semi-explicit in general. For instance, the circuit given in Figure~\ref{fig:Figure5} yields the DAE system
\begin{equation}\label{eq_non_semi_explicit}
\begin{split}
(C_1 + C_2) \partial_t v_1 - C_2 \partial_t v_2 & = i_S  \\
C_2 \partial_t v_1 - C_2 \partial_t v_2 & = \frac{1}{R} v_2 ,
\end{split}
\end{equation}
which is not semi-explicit because it has two differential variables on the left-hand side.

Extensions of this work are needed to tackle non-linear DAE systems arising from non-linear electronic components such as diodes and transistors. The conversion of polynomial ODEs to Hungarian and positive ones, and thus to CNRs, works essentially unchanged also for non-linear polynomial systems, and further extends to ODE systems including trigonometric and exponential functions, which can model transistors. However, this must be coupled with a general method for converting non-linear DAEs arising from electronic circuits to ODEs.

\bibliographystyle{IEEEtran}
\bibliography{rsi}

\appendix

\begin{proof}[Proof of Proposition~\ref{prop_positivation}]
Since $\hA^+,\hA^-,\hb^+$ and $\hb^-$ are non-negative, the theory of differential inequalities (or monotonic systems) readily implies that the solution $(x^+,x^-)$ of~(\ref{eq_positivation}) is non-negative whenever $(x^+(0), x^-(0))$ is non-negative. To see the second statement, let $(x^+,x^-)$ solve~(\ref{eq_positivation}) for $x(0) = x^+(0) - x^-(0)$. Then
\begin{align*}
\partial_t{x}^+ - \partial_t{x}^- & = (\hA^+ x^+  +   \hA^- x^-  +  \hb^+) \\
& \qquad - (\hA^+ x^-  +   \hA^- x^+  +  \hb^-) \\
& = (\hA^+ - \hA^-)(x^+ - x^-) + (\hb^+ - \hb^-) \\
& = \hA (x^+ - x^-) + \hb
\end{align*}
Since $\partial_t{x} = \hA x + \hb$ admits a unique solution, it must hold $x = x^+ - x^-$.
\end{proof}

\begin{proof}[Proof of Proposition~\ref{prop_bounded}]
Write~(\ref{eq_hungarization}) as $\partial_t{x}^+_i = g^+_i(x^+,x^-)$ and $\partial_t{x}^-_i = g^-_i(x^+,x^-)$. Since $g^+_i(z^+,z^-) \geq - \gamma z_i^+ z_i^-$ and $g^-_i(z^+,z^-) \geq - \gamma z_i^+ z_i^-$ for any $(z^+,z^-) \in \RE^{2n}_{\geq0}$ and the ODE system
\begin{align*}
\partial_t{z}^+_i & = - \gamma z_i^+ z_i^-, &
\partial_t{z}^-_i & = - \gamma z_i^+ z_i^-,
& 1 \leq i \leq n
\end{align*}
remains non-negative if initialized with non-negative values, we conclude that $(x^+,x^-)$ remains non-negative. Moreover, since $\gamma Q$ is non-positive on $\RE^{2n}_{\geq0}$, the solution of~(\ref{eq_hungarization}) is defined on $[0;\infty)$ and does not exhibit a finite explosion time. Since the second claim follows trivially from Proposition~\ref{prop_positivation} because the $Q$ terms cancel each other out in $\partial_t{x}^+ - \partial_t{x}^-$, let us focus on the third claim and set $\xi := \sup_{0 \leq t \leq \infty} \lVert x(t) \rVert_\infty < \infty$. Note that $x_i = \xp_i - \xm_i$, hence we get
\begin{align*}
\partial_t{x}^+_i & = \big(\hA^+ x^+  +   \hA^- x^-  +  \hb^+\big)_i - \gamma \xp_i \xm_i \\
& = \big(\hA^+ x^+  +   \hA^- x^-  +  \hb^+\big)_i - \gamma \xp_i (\xp_i - x_i) \\
& \leq \big(\hA^+ x^+  +   \hA^- x^-  +  \hb^+\big)_i + \gamma \xi \xp_i - \gamma (\xp_i)^2
\end{align*}
Since a similar calculation implies that
\begin{align*}
\partial_t{x}^-_i \leq \big(\hA^+ x^-  +   \hA^- x^+  +  \hb^-\big) +  \gamma \xi \xm_i - \gamma (\xm_i)^2 ,
\end{align*}
we infer that there exists a $\zeta > 0$ such that, for all $i$ and $(\zp,\zm) \in \RE^{2n}_{\geq0}$, it holds that
\begin{itemize}
    \item $g^+_i(\zp,\zm) \leq -1$ if $\zp_i \geq \zeta$ and;
    \item $g^-_i(\zp,\zm) \leq -1$ when $\zm_i \geq \zeta$.
\end{itemize}
This ensures that for any initial condition $(x^+(0),x^-(0)) \in \RE^{2n}_{\geq0}$, the solution $(\xp,\xm)$ enters eventually $[0;\zeta]^{2n}$ in order to remain there forever.
\end{proof}

\begin{proof}[Proof of Proposition~\ref{prop_crn_encoding}]
Straightforward.
\end{proof}

\begin{proof}[Proof of Theorem~\ref{thm_main_ext_ext}]
Follows from a direct combination of Proposition~\ref{prop_crn_encoding} and Theorem~\ref{thm_main_ext}.
\end{proof}

\section*{Proof of Theorem~\ref{thm_main} and Theorem~\ref{thm_main_ext}}

Before proving Theorem~\ref{thm_main} and~\ref{thm_main_ext}, we first have to establish some auxiliary results. To allow for a compact notation, we denote in the present section the $i$-th step of the numeric sequence by $x_i$ rather than $x[i]$.

\begin{proposition}\label{prop_0}
Consider the ODE systems $\partial_t{x} = F(x)$ and $\partial_t{x}_h = F_h(x)$ where $F$ and $F_h$ are assumed to be Lipschitz continuous on some bounded domain $B \subseteq \RE^n$ and $L \geq 0$ denotes the Lipschitz constant of $F$. Let us assume further that both ODE systems have solutions on $[0;T]$ which remain in $B$ and that $\sup\{ \norm{F(x) - F_h(x)} \mid x \in B \} \leq \eta$. Then, if $x(0) = x_h(0)$, for all $0 \leq t \leq T$ it holds that
\[
\norm{x(t)-x_h(t)} \leq \frac{\eta}{L} ( e^{L t} - 1 )
\]
\end{proposition}

\begin{proof}
We first show a modified version of Gronwall's inequality. To be more specific, let $\xi_1$ and $\xi_2$ be positive constants and $v$ a continuous function on $0 \leq t < \infty$ such that
\begin{align}\label{prop_0_aux}
v(t) \leq \xi_2 t + \xi_1 \int_0^t v(s) ds
\end{align}
Then, it holds that $v(t) \leq \tfrac{\xi_2}{\xi_1} (\xe^{\xi_1 t} - 1)$. To see this, we first rewrite~(\ref{prop_0_aux}) to
\[
v(t) + \frac{\xi_2}{\xi_1} \leq \frac{\xi_2}{\xi_1} + \xi_1 \int_0^t \left(v(s)+\frac{\xi_2}{\xi_1} \right) ds
\]
Since this rewrites to $\tilde{v}(t) \leq \tilde{\alpha} + \int_0^t \tilde{v}(s) \tilde{w}(s) ds$ for $\tilde{v}(s) := v(s) + \tfrac{\xi_2}{\xi_1}$, $\tilde{\alpha} := \tfrac{\xi_2}{\xi_1}$ and $\tilde{w}(s) := \xi_1$, Gronwall's inequality ensures that $\tilde{v}(t) \leq \tilde{\alpha} \cdot \xe^{\int_0^t \tilde{w}(s) ds}$ and we infer the auxiliary statement. This, in turn, yields
\begin{align*}
\norm{x(t) - x_h(t)} & \leq \norm{x(0) - x_h(0)} \\
& \qquad + \norm{\int_0^t \Big( F(x(s)) - F_h(x_h(s)) \Big) ds} \\
& \leq \norm{\int_0^t \big(F(x(s)) - F(x_h(s))\big) ds} \\
& \qquad + \norm{\int_0^t \big(F(x_h(s)) - F_h(x_h(s))\big) ds} \\
& \leq L \int_0^t \norm{x(s) - x_h(s)} ds + \eta t \\
& \leq \frac{\eta}{L} ( e^{L t} - 1 )
\end{align*}
\end{proof}

\begin{proposition}\label{prop_1}
Let $E \partial_t{x} = A x + b$ be a regular linear DAE system and let $\calD \subseteq \RE^n$ denote the corresponding set of consistent initial conditions. Then, $\calD$ is an affine subspace of $\RE^n$ and $x + h F_h(x) \in \calD$ whenever $x \in \calD$.
\end{proposition}

\begin{proof}
To see that $\calD$ is an affine subspace of $\RE^n$, please refer to~\cite[Section 2.1]{KunkelMehrmann}. Note further that $x_i = x_{i-1} + h F_h(x_{i-1})$ defines the backward Euler scheme which is applied to the DAE system $E \partial_t{x} = A x + b$, see~\cite[Section 5.2]{KunkelMehrmann}. Consider the BDF-1 scheme~\cite[Section 5.3]{KunkelMehrmann} which is given by
\[
\tfrac{1}{h} E (x_i - x_{i-1}) = A x_i + b
\]
if applied to $E \partial_t{x} = A x + b$. With this, we first observe that
\begin{align*}
\phantom{\Leftrightarrow} \qquad & &  \tfrac{1}{h} E (x_i - x_{i-1}) & = A x_i + b \\
\Leftrightarrow \qquad & &  \tfrac{1}{h} E x_i - A x_i & = \tfrac{1}{h} E x_{i-1} + b \\
\Leftrightarrow \qquad & & (\tfrac{1}{h} E - A) x_i & = \tfrac{1}{h} E x_{i-1} + b \\
\Leftrightarrow \qquad & & (E - h A) x_i & = E x_{i-1} + h b \\
\Leftrightarrow \qquad & & x_i & = (E - h A)^{-1} (E x_{i-1} + h b) ,
\end{align*}
where the inversion in the last line can always be performed for sufficiently small $h$ because $E \partial_t{x} = A x + b$ is regular. This, in turn, yields
\begin{align*}
& \frac{x_i - x_{i-1}}{h} = (E - h A)^{-1} b + \big( (E - h A)^{-1} E - I \big) \tfrac{1}{h} x_{i-1} \\
& \quad = (E - h A)^{-1} b + (E - h A)^{-1} (E - (E - h A) ) \tfrac{1}{h} x_{i-1} \\
& \quad = (E - h A)^{-1} b + (E - h A)^{-1} A x_{i-1} \\
& \quad = (E - h A)^{-1} (A x_{i-1} + b) \\
& \quad = F_h(x_{i-1})
\end{align*}
This shows that the backward Euler scheme and the BDF-1 scheme are identical if applied to $E \partial_t{x} = A x + b$. With this, the statement of the proposition is closely related to~\cite[Remark 5.25]{KunkelMehrmann}. To see it, we may assume without loss of generality (see proof of~\cite[Theorem 5.24]{KunkelMehrmann}) that $E \partial_t{x} = A x + b$ is such that $A = I$ and $E = N$ for some nilpotent $N$ with $N^\nu = 0$ and $N^{\nu - 1} \neq 0$. It can be easily seen that in such a case the solution is $x \equiv -b$, thus implying in particular that the set of consistent initial conditions is $\calD = \{-b\}$. Moreover, the BDF-1 scheme rewrites to
\begin{align*}
\phantom{\Leftrightarrow} \qquad & & (\tfrac{1}{h} N - I) x_i & = \tfrac{1}{h} N x_{i-1} + b \\
\Leftrightarrow \qquad & & (I - \tfrac{1}{h} N) x_i & = -\tfrac{1}{h} N x_{i-1} - b \\
\Leftrightarrow \qquad & & x_i & = -(I - \tfrac{1}{h} N)^{-1} (\tfrac{1}{h} N x_{i-1} + b) \\
\Leftrightarrow \qquad & & x_i & = - \sum_{l = 0}^{\nu - 1} (\tfrac{1}{h} N)^l (\tfrac{1}{h} N x_{i-1} + b) ,
\end{align*}
where the last equivalence is due to the Neumann series and the nilpotency of $N$. This, in turn, implies that
\begin{align*}
x_i & = - \sum_{l = 0}^{\nu - 1} (\tfrac{1}{h} N)^l (\tfrac{1}{h} N x_{i-1} + b) \\
& = - \sum_{l = 1}^{\nu - 1} (\tfrac{1}{h} N)^l x_{i-1} - \sum_{l = 0}^{\nu - 1} (\tfrac{1}{h} N)^l b \\
& = - b - \sum_{l = 1}^{\nu - 1} (\tfrac{1}{h} N)^l (x_{i-1} + b) ,
\end{align*}
thus showing that $x_i = -b$ whenever $x_{i-1} = -b$.
\end{proof}

\begin{proposition}\label{prop_2}
Let $E \partial_t{x} = A x + b$ be a regular linear DAE system and let $\calD \subseteq \RE^n$ denote the corresponding set of consistent initial conditions. Then
\begin{itemize}
    \item The solution of $E \partial_t{x} = A x + b$ is contained in $\calD$.
    \item There exist $\hat{A} \in \RE^{n \times n}$ and $\hat{b} \in \RE^n$ such that the solution of the ODE system $\partial_t{x} = \hat{A} x + \hat{b}$ coincides with that of $E \partial_t{x} = A x + b$ for all $x(0) \in \calD$.
    \item Together with $F_h(x) := (E - hA)^{-1}(A x + b)$, where $h > 0$, it holds that $F_h$ converges uniformly, as $h \to 0$, to $\hat{A} x + \hat{b}$ on any bounded subset of $\calD$.
\end{itemize}
\end{proposition}

\begin{proof}
The first two points are well-known in the theory of linear DAE systems, see Section~\cite[Section 2.1]{KunkelMehrmann} (it is interesting to note that an efficient computation of $\hat{A} \in \RE^{n \times n}$ and $\hat{b} \in \RE^n$ is difficult because it relies on index reduction~\cite{doi:10.1137/0909014}).

To see third claim, we observe that $x_i = x_{i-1} + h F_h(x_{i-1})$ defines the backward Euler scheme applied to the DAE system $E \partial_t{x} = A x + b$, see~\cite[Section 5.2]{KunkelMehrmann}. We next show that $x_0 \mapsto \frac{1}{h}(x_1 - x_0)$ converges uniformly on any bounded subset of $\calD$ to $x_0 \mapsto \hat{A} x_0 + \hat{b}$ when $h \to 0$. To this end, we may assume without loss of generality (see discussion after Equation 5.25 in~\cite{KunkelMehrmann}) that the DAE system $E \partial_t{x} = A x + b$ is such that
\[
E = \left(
\begin{array}{c|c}
I & 0 \\ \hline
0 & N
\end{array}
\right) \text{ and }
A = \left(
\begin{array}{c|c}
J & 0 \\ \hline
0 & I
\end{array}
\right) ,
\]
where $N$ is such that $N^\nu = 0$ and $N^{\nu-1} \neq 0$ for some $\nu \geq 1$. This implies that the solution of $E \partial_t{x} = A x + b$ is characterized by a pair of \emph{decoupled} dynamical systems, namely by the ODE system $\partial_t{x}^\rI = J x^\rI + b^\rI$ and the DAE system $N \partial_t{x}^\rII = x^\rII + b^\rII$, where $x = (x^\rI, x^\rII)$ and $b =(b^\rI, b^\rII)$. Thanks to this, it suffices to consider $x^\rI_1 - x^\rI_0$ and $x^\rII_1 - x^\rII_0$ separately.

Since $x_\rII \equiv - b_\rII$ solves $N \partial_t{x}^\rII = x^\rII + b^\rII$, we infer that $\calD = \{ (x^\rI, x^\rII) \mid x^\rII = - b^\rII \}$. Hence, Proposition~\ref{prop_1} shows that $x^\rII_1 - x^\rII_0 = 0$ whenever $x_0 \in \calD$.

We next focus on $x^\rI_1 - x^\rI_0$. Thanks to the fact that $\partial_t{x}^\rI = J x^\rI + b^\rI$, we have to investigate the local truncation error of the backward Euler scheme in the context of a linear ODE system. Despite the fact this is discussed in many books about ODEs, we provide here a proof because most texts do not show that the local truncation error converges \emph{uniformly} to zero on arbitrarily large compact sets. To this end, we first observe that the Taylor expansion of $x^\rI$ around zero yields
\begin{align*}
x^\rI(h) = x^\rI_0 + (J x^\rI_0 + b^\rI) h + \ddot{x}^\rI(\xi) \tfrac{h^2}{2}
\end{align*}
for some $\xi \in (0;h)$. With $\tilde{F}_h(x^\rI_0) = (I - hJ)^{-1}(J x^\rI_0 + b^\rI)$, the proof of Proposition~\ref{prop_1} implies that $\tilde{F}_h(x^\rI_0) = \frac{1}{h}(x^\rI_1 - x^\rI_0)$. This, in turn, implies that
\begin{align*}
x^\rI(h) - x^\rI_1 & = x^\rI(h) - (x^\rI_0 + h \tilde{F}_h(x^\rI_0)) \\
& = x^\rI_0 + (J x^\rI_0 + b^\rI) h + \ddot{x}^\rI(\xi) \tfrac{h^2}{2} \\
& \quad - [x^\rI_0 + h(I - h J)^{-1}(J^\rI x_0 + b^\rI)] \\
& = h^2 [\tfrac{1}{2}\ddot{x}^\rI(\xi) + \tfrac{1}{h} (I - (I - h J)^{-1})(J x^\rI_0 + b^\rI)]
\end{align*}
In the case $h \leq 1 / (2 \norm{J})$, the Neumann series allows us to deduce that
\begin{align*}
I - (I - h J)^{-1} & = (I - h J)(I - h J)^{-1} - (I - h J)^{-1} \\
& = ((I - h J) - I)(I - h J)^{-1} \\
& = - h J (I - h J)^{-1} \\
& = - h J \sum_{k=0}^\infty (h J)^k
\end{align*}
with $\norm{\sum_{k=0}^\infty (h J)^k} \leq \sum_{k = 0}^\infty 2^{-k} = 2$. Moreover, a differentiation of $\partial_t{x}^\rI = J x^\rI + b^\rI$ yields $\ddot{x}^\rI = J^2 x^\rI + J b^\rI$. This and the last statement imply the existence of constants $\zeta_1, \zeta_2 \geq 0$ that neither depend on $x^\rI_0$ nor on $h$ and that satisfy
\begin{align*}
\norm{x^\rI(h) - x^\rI_1} & \leq h^2 \big(\zeta_1 + \zeta_2 \norm{x^\rI_0}\big)
\end{align*}
for all $0 \leq h \leq 1$. This shows that $x^\rI_0 \mapsto \frac{1}{h}(x^\rI_1 - x^\rI_0)$ converges uniformly on any bounded set to $x^\rI_0 \mapsto J x^\rI_0 + b^\rI$.
\end{proof}

We are in a position to prove Theorem~\ref{thm_main}.

\begin{proof}[Proof of Theorem~\ref{thm_main}]
Let $\hat{A} \in \RE^{n \times n}$ and $\hat{b} \in \RE^n$ be as in Proposition~\ref{prop_2} and fix $T > 0$ and $x(0) \in \calD$. Since the solution of $E \partial_t{x} = A x + b$ solves the linear ODE system $\partial_t{x} = \hat{A} x + \hat{b}$, this implies that $x$ exists and is bounded on $[0;T]$. Hence, there exists a closed ball $B_\rho(0)$ centered at $0 \in \RE^n$ with radius $\rho > 0$ such that $x(t) \in B_\rho(0)$ for all $0 \leq t \leq T$. Since $B_\rho(0)$ is bounded, Proposition~\ref{prop_2} ensures that $x$ is contained in $\calD$ and that for any $\eta > 0$ there exists an $h > 0$ such that
\[
\sup_{x \in B_\rho(0) \cap \calD} \norm{ \hat{A} x + \hat{b} - F_h(x) } \leq \eta
\]
Moreover, Proposition~\ref{prop_1} ensures that the solution $x_h$ of $\partial_t{x}_h = F(x_h)$ is contained in $\calD$. By combining the foregoing statements, Proposition~\ref{prop_0} yields the claim.
\end{proof}

The following auxiliary results are needed for the proof of Theorem~\ref{thm_main_ext}

\begin{proposition}\label{prop_extended_case}
Fix $E, A \in \RE^{n \times n}$, $B \in \RE^{n \times (k + m)}$, $D \in \RE^{(k + m) \times (k + m)}$, $d \in \RE^{(k + m)}$ and consider the linear DAE system
\[
\underbrace{
\left(
\begin{array}{c|c}
E & 0 \\ \hline
0 & I
\end{array}
\right)}_{\hE :=}
\underbrace{
\left(
\begin{array}{c}
\partial_t{x} \\ \hline
\partial_t{u}
\end{array}
\right)}_{ \partial_t{\hx} :=}
=
\underbrace{
\left(
\begin{array}{c|c}
A & B \\ \hline
0 & D
\end{array}
\right)}_{\hA :=}
\underbrace{
\left(
\begin{array}{c}
x \\ \hline
u
\end{array}
\right)}_{\hx :=}
+
\underbrace{
\left(
\begin{array}{c}
0 \\ \hline
d
\end{array}
\right)}_{\hb :=}
\]
Then, $(\hE - h \hA)^{-1} (\hA \hx + \hb)$ is given by
\begin{align}\label{eq_inv_expression}
\left(
\begin{array}{c}
(E - hA)^{-1} \big( Ax + B u + h B (I - hD)^{-1} (D u + d) \big) \\
(I - hD)^{-1} (D u + d)
\end{array}
\right)
\end{align}
\end{proposition}

\begin{proof}
By relying on the inversion formula for block matrices, we obtain
\begin{align*}
& (\hE - h \hA)^{-1} \hA \\
& = \left(
\begin{array}{c|c}
E - hA & -hB \\ \hline
0 & I - hD
\end{array}
\right)^{-1}
\left(
\begin{array}{c|c}
A & B \\ \hline
0 & D
\end{array}
\right) \\
& = \left(
\begin{array}{c|c}
(E - hA)^{-1} & (E - hA)^{-1} hB (I - hD)^{-1} \\ \hline
0 & (I - hD)^{-1}
\end{array}
\right)
\left(
\begin{array}{c|c}
A & B \\ \hline
0 & D
\end{array}
\right) \\
& = \left(
\begin{array}{c|c}
(E - hA)^{-1} A & (E - hA)^{-1} ( B + h B (I - hD)^{-1} D ) \\ \hline
0 & (I - hD)^{-1} D
\end{array}
\right)
\end{align*}
Armed with this, we infer that
\begin{align*}
& (\hE - h \hA)^{-1} \hA
\left(
\begin{array}{c}
x \\ \hline
u
\end{array}
\right) \\
& =
\left(
\begin{array}{c}
(E - hA)^{-1} \big( Ax + B u + h B (I - hD)^{-1} D u \big) \\
(I - hD)^{-1} D u
\end{array}
\right)
\end{align*}
and
\begin{align*}
& (\hE - h \hA)^{-1}
\left(
\begin{array}{c}
0 \\ \hline
d
\end{array}
\right)
& =
\left(
\begin{array}{c}
(E - hA)^{-1} h B (I - hD)^{-1} d \\
(I - hD)^{-1} d
\end{array}
\right)
\end{align*}
A summation of the foregoing statements yields~(\ref{eq_inv_expression}).
\end{proof}

\begin{corollary}\label{thm_extended_case_2}
Fix an arbitrary consistent initial condition $(x(0),u(0))^T \in \RE^{n + k + m}$ of the DAE system from Proposition~\ref{prop_extended_case}. The corresponding ODE approximation is then
\begin{align*}
\partial_t{x}_h & = (E - hA)^{-1} \big( Ax_h + B u_h^{\langle 0 \rangle} + h B u_h^{\langle 1 \rangle} \big) \\
\partial_t{u}_h^{\langle 0 \rangle} & = (I - hD)^{-1} D u_h^{\langle 0 \rangle} + (I - hD)^{-1} d \\
\partial_t{u}_h^{\langle 1 \rangle} & = (I - hD)^{-1} D u_h^{\langle 1 \rangle}
\end{align*}
with $u^{\langle 0 \rangle}_h(0) = u(0)$ and $u^{\langle 1 \rangle}_h(0) = (I - hD)^{-1} D u(0)$.
\end{corollary}

\begin{proof}
Follows directly from Proposition~\ref{prop_extended_case}.
\end{proof}

\begin{proof}[Proof of Theorem~\ref{thm_main_ext}]
In Theorem~\ref{thm_main_ext}, replace $u$ with $\hat{u}$, $z$ with $\hat{z}$ and $B$ with $\hat{B}$. Afterwards, apply Corollary~\ref{thm_extended_case_2} to the case where $u := \binom{\hat{u}}{\hat{z}} \in \RE^{k + m}$ and
\[
B := \left(
\begin{array}{c|c}
\hat{B} & 0 \\ \hline
0 & 0
\end{array}
\right) \in \RE^{(k + m) \times (k + m)}
\]
\end{proof}

\end{document}